\documentclass[a4paper,twocolumn,english,showpacs,nofootinbib,nobibnotes,aps,pra]{revtex4-1}
\usepackage[utf8]{inputenc}
\usepackage[T1]{fontenc} % necessary for "old text encoding" apparently (like for \k{a} (ogonek used in Polish))
\usepackage{amsmath}
\usepackage{amsfonts}
\usepackage{amssymb}
\usepackage{amsthm}
\usepackage{braket}
\usepackage{dsfont}

\usepackage{graphicx}

\newtheorem{definition}{Definition}
\newtheorem{lemma}[definition]{Lemma}
\newtheorem*{lemma*}{Lemma}
\newtheorem{theorem}[definition]{Theorem}
\newtheorem{corollary}[definition]{Corollary}

\global\long\def\one{\mathds{1}}
\global\long\def\trace{\operatorname{Tr}}
\global\long\def\ketbra#1#2{\ket{#1}\!\bra{#2}}

\begin{document}

\title{Certifying quantum memories with coherence}

\author{Timo Simnacher}

\affiliation{Naturwissenschaftlich-Technische Fakult{\"a}t, 
Universit{\"a}t Siegen, Walter-Flex-Stra{\ss}e~3, 57068 Siegen, 
Germany}

\author{Nikolai Wyderka}

\affiliation{Naturwissenschaftlich-Technische Fakult{\"a}t, 
Universit{\"a}t Siegen, Walter-Flex-Stra{\ss}e~3, 57068 Siegen, 
Germany}

\author{Cornelia Spee}

\affiliation{Naturwissenschaftlich-Technische Fakult{\"a}t, 
Universit{\"a}t Siegen, Walter-Flex-Stra{\ss}e~3, 57068 Siegen, 
Germany}

\author{Xiao-Dong Yu}

\affiliation{Naturwissenschaftlich-Technische Fakult{\"a}t, 
Universit{\"a}t Siegen, Walter-Flex-Stra{\ss}e~3, 57068 Siegen, 
Germany}

\author{Otfried G\"uhne}

\affiliation{Naturwissenschaftlich-Technische Fakult{\"a}t, 
Universit{\"a}t Siegen, Walter-Flex-Stra{\ss}e~3, 57068 Siegen, 
Germany}

\date{\today}

\begin{abstract}
Quantum memories are an important building block for quantum information processing. 
Ideally, these memories preserve the quantum properties of the input.
We present general criteria for measures to evaluate the quality of quantum memories. Then, we
introduce a quality measure based on coherence satisfying these criteria,
which we characterize in detail for the qubit case. The measure can be estimated from sparse
experimental data and may be generalized to characterize other building blocks,
such as quantum gates and teleportation schemes.
\end{abstract}

\maketitle
%-------------------------------------------------------------------------------------%
\section{Introduction}
%-------------------------------------------------------------------------------------%
In order to work, quantum computers need reliable and well-characterized 
routines and devices. The loss of quantum coherence, however, is one of 
the major obstacles on the way to a scalable platform for quantum computing, 
and the suppression of decoherence is known as one of the DiVincenzo criteria 
for quantum computers \cite{diVincenzo2000}. One main ingredient in any 
computing architecture is the memory. Quantum computers are no exception 
and furthermore, quantum memories play a central role in the development 
of quantum repeaters \cite{simon2010quantumreview, briegelrepeater, duanrepeater}. Consequently, the search for reliable 
systems that store quantum states for a reasonable amount of time while 
preserving quantum properties is an active area of research 
\cite{julsgaard2004, choi2008, zhao2009, hedges2010, polzikexperiment, 
duanexperiment, guoexperiment}.

A possible way to verify the proper functioning of quantum gates and 
quantum memories is to completely characterize their behavior via quantum 
process tomography \cite{nielsentomography, tomographyresource}. This, 
however, requires an effort exponentially increasing in the size of the system. 
More importantly, it is desirable to determine the change of physical 
properties, such as entanglement and coherence, under the prescribed time evolution 
since these convey the quantumness of the underlying process.
By contrast, a complete characterization does not distinguish between these characteristics and
minor details, making it harder to identify the main features.
Therefore, it is beneficial to describe devices directly by their effect on physical phenomena. 

Several methods have been suggested to characterize quantum memories: 
The quantumness of channels has been assessed based on whether or not they preserve entanglement, focusing on reducing the number of measurements in bipartite optical systems \cite{haeselerlutkenhaus, haeselerlutkenhaus1}. 
Furthermore, quantum steering has been considered as a way to evaluate the performance of quantum channels in the case of untrusted measurement devices, again distinguishing channels that do and do not preserve entanglement \cite{pusey}. 
Finally, a resource theory of quantum memories has been developed \cite{rosset2018resource}. 
The free resources are channels that do not preserve entanglement. 
Using arbitrary pre- and postprocessing accompanied by unlimited classical memory as free operations, the authors establish a game-theoretic way to assess quantum memory performance based on the entanglement of the corresponding Choi state.
Nonetheless, these attempts require either well characterized test states as 
inputs, many measurements on the output or an advanced scheme that has to be implemented.

First conditions on how to generally assess the performance of quantum memories 
were discussed in Ref.~\cite{simon2010quantumreview}. That work suggests 
using the fidelity as a performance measure.
In fact, instead of the fidelity, any distance measure between the input- and the output state would
be suitable to measure the performance of such devices, e.g.~ a measure based on the coherence of the states \cite{hu2017relative}.
As the authors of Ref.~\cite{simon2010quantumreview} note, however, 
the fidelity is sensitive to unitary transformations of the input,
which may be compensated by the quantum computer controlling the interface.
With this in mind, the authors propose to use the purity of the memory
instead, which is indeed insensitive to unitary transformations. However,
the purity of a channel yielding a fixed, pure state independent of the 
input is maximal, but such a channel would certainly not qualify as 
a proper memory.

With these considerations in mind, we introduce  general criteria for 
quality measures of quantum memories. First, they should clearly distinguish 
schemes that require storing only classical information from perfect unitary 
transformations. Second, as we assume that unitary transformations can be 
corrected by the underlying quantum computer, the quality of a quantum memory
should be invariant under such unitary transformations.

We then propose a measure that obeys these natural properties
using the phenomenon of coherence. The key idea is that an ideal quantum memory
preserves the coherence in any basis. The measure can be used to
prove that a memory preserves entanglement, moreover, it can be estimated with
few measurements, without the need for well-characterized input states. 
Our concept may be generalized to characterize also other quantum primitives 
such as teleportation schemes and, using generalized notions of coherence 
\cite{ringbauerpiani, kraftpiani}, also to multi-particle quantum gates.

%-------------------------------------------------------------------------------------%
\section{Memory quality measures}
%-------------------------------------------------------------------------------------%
To start, let us study what physical properties a measure for the quality
of a quantum memory should have. As non-classical properties 
are essential for many quantum algorithms, the storage should 
preserve as many of these properties as possible. A perfect quantum memory 
is given by the identity channel. In practice, however, this is rather difficult 
to achieve. In contrast to that, measure-and-prepare (M\&P) schemes (also known as entanglement-breaking 
channels) can be easily simulated using only classical storage. One
just performs measurements on the input state and stores the result. Based
on that, one then prepares a quantum state on demand.

These two examples show that a measure for the quality of
a quantum memory should have two natural properties: First, 
it should be maximal for memories that preserve the input state perfectly. As we 
assume that we can perform unitary rotations, we also allow the memory to 
apply a known and fixed unitary rotation to the input. Second, the measure 
should have a non-maximal quality for the M\&P schemes 
described above, certifying genuine quantum storage.

Formally, M\&P channels can be written as 
\cite{entbrchan}
\begin{align}
    \mathcal{M}(\rho) = \sum_{\lambda} \trace(E_\lambda \rho)\rho_\lambda,
\end{align}
where the set $\{E_\lambda\}$ forms a positive operator valued measure
and the $\rho_\lambda$ are density matrices. In the following, 
let $\mathcal{M}$ be a quantum channel, i.e.,~a completely 
positive, trace preserving map \cite{bengtsson2017geometry}. We can now formulate our criteria for quality
measures.

\begin{definition} \label{defmap}
A map $Q(\mathcal{M}) \in [0,1]$ for a channel $\mathcal{M}$ is called 
\emph{memory quality measure}, if it satisfies the following.
\begin{enumerate}
    \item[M1:] $Q(\mathcal{M})=1$ if $\mathcal{M}(\rho) = V \rho V^\dagger$ for
    some unitary $V$,
    \item[M2:] $Q(\mathcal{M})\leq c$ for some constant $c \in [0,1)$ if 
    $\mathcal{M}$ is an M\&P channel.
\end{enumerate}
A memory quality measure is called \emph{sharp}, if it additionally fulfills the following.
\begin{enumerate}
    \item[M1':] $Q(\mathcal{M})=1 \Leftrightarrow \mathcal{M}(\rho) = V \rho V^\dagger$ for some unitary $V$.
\end{enumerate}
\end{definition}

Obviously, condition M1 implies that the identity channel has unit quality.
Furthermore, for continuous sharp measures, M1' implies M2 since M\&P channels 
have a finite distance to the set of unitary channels due to the compactness of the set \cite{bengtsson2017geometry}.

%-------------------------------------------------------------------------------------%
\section{Definition of the measures}
%-------------------------------------------------------------------------------------%
Recently, there has been growing interest in coherence in the light of resource 
theories \cite{baumgratzCoherence}. This has led to the development of various 
coherence measures that quantify the amount of coherence present in a given $D$-dimensional state.
For a fixed basis (defined by some unitary 
$U$ such that $\ket{b_i}:=U\ket{i}$), we use the normalized 
robustness of coherence \cite{napoli2016robustness}
\begin{align}
    C_U(\rho) := \frac1{D-1} \min_{\tau \in \mathcal{D}} \left\{ s \geq 0 \middle\vert \frac{\rho + s \tau}{1 + s} \in \mathcal{I}_U \right\},
\end{align}
where $\mathcal{D}$ is the set of all $D$-dimensional states and $\mathcal{I}_U$ is the set of incoherent (i.e., diagonal) $D$-dimensional states with regard to~the basis $U\ket{j}$.
However, our results are valid for any continuous and convex coherence 
measure with the property that the only states maximizing the measure for 
a fixed basis $U$ are given by
\begin{align}\label{eq:maxcoh}
\ket{\Psi_U^{\vec{\alpha}}} 
:=  \frac{1}{\sqrt{D}} \sum_{j=0}^{D-1} e^{i \alpha_j} \ket{b_j}
= U Z_{\vec{\alpha}} \ket{+},
\end{align}
where $\vec{\alpha}$ is some $D$-dimensional vector of phases and 
$Z_{\vec{\alpha}}$ is a diagonal unitary matrix with entries $e^{i\alpha_j}$, 
acting on $\ket{+} := \frac{1}{\sqrt{D}}\sum_i \ket{i}$. 
Note that the states in Eq.~(\ref{eq:maxcoh}) maximize any valid coherence 
monotone, and for many prominent coherence measures 
such as the robustness of coherence \cite{piani2016robustness}, the $l_1$-norm of coherence \cite{peng2016maximally}, and the relative entropy measure \cite{bai2015maximally}, they are the only states doing so. Furthermore, they are also maximally coherent in a 
resource theoretic sense \cite{baumgratzCoherence, peng2016maximally}.

We define a physically motivated quality measure from the following 
considerations: Given a quantum channel $\mathcal{M}$, there is a 
``most classical'' basis, in which even the most robust maximally 
coherent state with respect to that basis is mapped to a state with
small coherence. This basis is identified by our proposed measure, 
and the conserved coherence in this basis defines the quality.

\begin{definition} 
\label{qualities}
For a quantum channel $\mathcal{M}$, the quality $Q_0$ is given by
\begin{align}
    Q_0(\mathcal{M}) := \min_U \max_{\vec{\alpha}}  C_U[\mathcal{M}(\ket{\Psi_U^{\vec{\alpha}}})].
\end{align}
\end{definition}
Here, we write $\mathcal{M}(\ket{\Psi_U^{\vec{\alpha}}})$ 
instead of $\mathcal{M}(\ketbra{\Psi_U^{\vec{\alpha}}}{\Psi_U^{\vec{\alpha}}})$ 
for convenience. 
If $Q_0(\mathcal{M})=1$, then in any basis at least one maximally coherent state is preserved. Later, we show that this already implies that $\mathcal{M}$ is unitary.

To give an operational interpretation of $Q_0$, we consider
a phase discrimination task. Here, the improvement of the success probability  over naive guessing using a quantum state $\rho$ is determined by its robustness of coherence \cite{napoli2016robustness}.
If a suitable maximally coherent state is stored before it is used as a probe state, the improvement is quantified by $Q_0$. Thus, $Q_0$ certifies how well a quantum memory preserves the usefulness of a maximally coherent state for a phase discrimination task, without specifying the incoherent basis $U$.

Despite the clear physical interpretation of this measure, there are related quantities which turn out to be useful for the discussion. Therefore, we introduce two additional parameters, which provide an upper and lower bound on $Q_0$.
First, we consider the minimal coherence left in any basis of the most robust maximally coherent states if one minimizes over their bases:

\begin{definition}
For a quantum channel $\mathcal{M}$, the quantity $Q_-$ is defined by
\begin{align}
    Q_-(\mathcal{M}) := \min_{U,U^\prime} \max_{\vec{\alpha}}  C_{U^\prime}[\mathcal{M}(\ket{\Psi_U^{\vec{\alpha}}})].
\end{align}
\end{definition}
In contrast to $Q_0$, the basis of coherence is varied independently of 
the basis of the maximally coherent states. Thus, we have that 
$Q_-(\mathcal{M}) \leq Q_0(\mathcal{M})$. Second, as an upper bound to $Q_0$, 
we consider the minimal coherence in any basis maximized over all states in the range:

\begin{definition}
\label{qualities2}
For a quantum channel $\mathcal{M}$, the quantity $Q_+$ is defined by
\begin{align}
    Q_+(\mathcal{M}) :=\, & \min_{U} \max_{\rho}  C_{U}[\mathcal{M}(\rho)] \nonumber \\
                      =\, & \min_{U} \max_{\ket{\psi}}  C_{U}[\mathcal{M}(\ket{\psi})],
\end{align}
where the equality is due to the convexity of the coherence measure and linearity of $\mathcal{M}$.
\end{definition}

Here, in contrast to $Q_0$, the maximization is not limited to maximally coherent 
states. Hence, it holds that
\begin{align}
    Q_-(\mathcal{M}) \leq Q_0(\mathcal{M}) \leq Q_+(\mathcal{M}).
\end{align}
Due to the minimization over all bases $U$ (and $U^\prime$ for $Q_-$), 
for all channels $\mathcal{M}$ and unitary channels $\mathcal{V}$ with 
$\mathcal{V}(\rho) = V\rho V^\dagger$ where $V$ is some unitary, we have 
the following identities:
\begin{align}
 Q_\pm(\mathcal{M}) &= Q_\pm(\mathcal{V}\circ \mathcal{M}) = Q_\pm(\mathcal{M}\circ \mathcal{V}), 
 \nonumber \\
 Q_0(\mathcal{M}) &= Q_0(\mathcal{V}\circ \mathcal{M} \circ \mathcal{V}^{-1}).
\end{align}
The quantities $Q_\pm$ are completely invariant under prior and subsequent rotations, whereas $Q_0$ is only invariant under joint rotations. As such, the quantities $Q_\pm$ are useful to obtain bounds on $Q_0$.

Note that all measures are continuous in the space of quantum channels (for the proof, see Appendix A).

%-------------------------------------------------------------------------------------%
\section{Properties of the measures}
%-------------------------------------------------------------------------------------%
We now show that the quantities $Q_0$ and $Q_\pm$
are sharp memory quality measures. First, using the Sinkhorn normal form of unitaries \cite{sinkhorn}, 
we show the following.
\begin{lemma}
The measures $Q_\pm$ and $Q_0$ fulfill property M1, i.e.~$Q(\mathcal{V})=1$ for 
all unitary channels $\mathcal{V}$.
\end{lemma}
\begin{proof}
As $Q_-(\mathcal{M})\leq Q_0(\mathcal{M})\leq Q_+(\mathcal{M})$, it suffices to 
show the property for $Q_-$. Furthermore, as $Q_-$ is invariant under unitary 
rotations, it suffices to consider only the identity channel $\operatorname{id}$.
Recall that
\begin{align}
    Q_-(\operatorname{id}) = \min_{U,U^\prime} \max_{\vec{\alpha}} C_{U^\prime}(UZ_{\vec{\alpha}}\ket{+}) = 1,
\end{align}
where $Z_{\vec{\alpha}}$ is a diagonal matrix with phases $e^{i\alpha_j}$ as entries, 
is equivalent to the statement that for all bases $U$ and $U^\prime$, there exists 
a maximally coherent state in $U$ that is also maximally coherent in $U^\prime$. 
This can be stated as follows: for all $U$ there exist vectors $\vec{\alpha}$ and $\vec{\beta}$, such that
\begin{align}
    Z_{\vec{\beta}}^\dagger UZ_{\vec{\alpha}}\ket{+} = \ket{+}, \label{coherent_intersection}
\end{align}
which is equivalent to the statement that the sets of maximally coherent states 
with regard to~two different bases always have a non-empty intersection. This interesting 
geometrical question has been investigated and answered positively recently; 
it was shown that any unitary operator $U$ can be decomposed as \cite{sinkhorn}
\begin{align}
    U = Z_1 X Z_2,
\end{align}
where $Z_1$ and $Z_2$ are diagonal unitaries with the upper left entry equal to 
1 and $X$ is a unitary matrix where the elements in each row and each column 
sum to 1. Inserting this decomposition into Eq.~(\ref{coherent_intersection}) 
shows that choosing $\vec{\alpha}$ and $\vec{\beta}$ such that 
$Z_{\vec{\alpha}} = Z_1^\dagger$ and $Z_{\vec{\beta}} = Z_2$ yield the desired 
equality, 
as $\ket{+}$ is an eigenstate of $X$.
\end{proof}

Second, also the converse statement holds. The proof is given in Appendix B.
\begin{theorem}
$Q_\pm$ and $Q_0$ fulfill property M1', i.e., if $Q(\mathcal{M})=1$, then 
$\mathcal{M}$ is a unitary channel.
\end{theorem}

Finally, as the continuity of $Q_\pm$ and $Q_0$ together with property M1' implies property M2, it follows that:
\begin{corollary}
 The quantities $Q_\pm$ and $Q_0$ are sharp memory quality measures.
\end{corollary}

In the case of single-qubit channels, we can find tight numerical bounds on the quality of M\&P channels (see Theorem \ref{theorem:qubitbounds} below). 

Additionally, the quality measure $Q_+$ satisfies a useful preprocessing property:

\begin{lemma}\label{lemma:prepqplus}
The quality measure $Q_+$ cannot be increased by preprocessing the input, i.e.~$Q_+(\mathcal{M}\circ \mathcal{N})\leq Q_+(\mathcal{M})$ for all quantum channels $\mathcal{M}$ and $\mathcal{N}$.
\end{lemma}
\begin{proof}
By definition, 
\begin{align}
Q_+(\mathcal{M}\circ \mathcal{N}) &=\min_U \max_{\rho}  C_U(\mathcal{M}(\mathcal{N}(\rho))) \nonumber \\
&\leq  \min_U \max_{\rho}  C_U(\mathcal{M}(\rho)) =Q_+(\mathcal{M}),
\end{align}
which proves the lemma.
\end{proof}

For $Q_-$, we can prove a similar statement for the case of unital, i.e., channels that map the maximally mixed state to itself, single-qubit 
channels (see Lemma~\ref{lemma:singlequbitunitalprep}).

The measures introduced in Ref.~\cite{rosset2018resource} are monotonous under pre- and postprocessing using unlimited classical memory and preexisting randomness. 
This is not true for $Q_0$ and $Q_\pm$. A counterexample is given by the channel $\mathcal{N}$ defined by $\vec{\lambda} = (0,0,1)$ and vanishing $\vec{\kappa}$, and the M\&P channel  $\mathcal{M}$ maximizing $Q_0$, given by $\vec{\lambda} = (0,0,\frac1{\sqrt2})$, $\vec{\kappa} = (\frac1{\sqrt2},0,0)$. Then, $Q_0(\mathcal{M}\circ \mathcal{N}) = \frac1{\sqrt2} \nleq Q_0(\mathcal{N}) = 0$. This counterexample also works for $Q_+$. For $Q_-$, choosing $\mathcal{N}$ as the planar channel with semiaxes $\vec{\lambda} = (0,\frac12, \frac12)$ and zero displacement, and $\mathcal{M}$ as the channel maximizing $Q_-$, i.e., defined by $\vec{\lambda} = (0,\frac1{\sqrt5}, \frac1{\sqrt5})$ and displacement $\vec{\kappa} = (\frac1{\sqrt5},0,0)$, leads to $Q_-(\mathcal{M}\circ \mathcal{N}) = \frac1{2\sqrt5}\nleq Q_-(\mathcal{N}) = 0$. The non-monotonicity is expected for measures based on coherence, because in contrast to entanglement, coherence can be created locally. Furthermore, if a measure is monotonous under the operations defined in Ref.~\cite{rosset2018resource}, it would assign the same quality to all M\&P channels. However, some M\&P channels are more useful than others for the task of phase discrimination.

It should be noted that the measures introduced here are not faithful in the sense that any non-M\&P channel can be detected. This is not possible with an efficiently computable single measure, because such a measure would solve the separability problem, which is NP-hard \cite{gharibian2010strong}.

%-------------------------------------------------------------------------------------%
\section{The single-qubit case}
%-------------------------------------------------------------------------------------%
\begin{figure}[t]
 \centering
 \includegraphics[width=0.65\columnwidth]{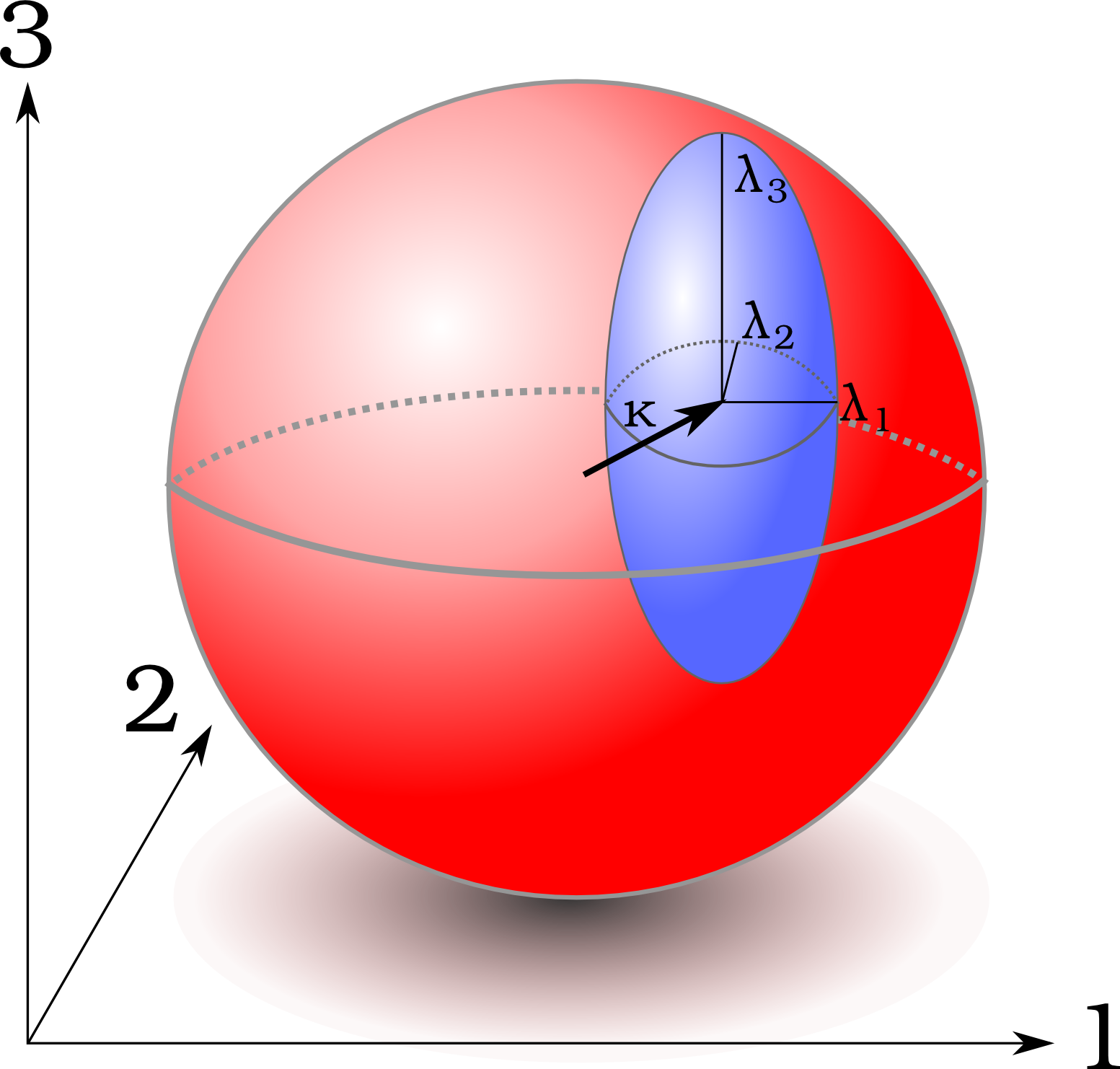}
 \caption{The image of the Bloch sphere (red area) of single-qubit maps is an ellipsoid (blue area) with semiaxes $\lambda_i$, displaced by $\vec{\kappa}$.}
 \label{fig:qubitmaps}
\end{figure}
%%%%%%%%%%%%%%%%%%%%%%%%%%%%%%%%%%%%%%%%%%%%%%%%%%%%%%%%%
The action of single-qubit channels can be well understood in the Bloch picture.
The Bloch decomposition of a qubit state is given by $\rho = \frac{1}{2} (\one + \vec{v} \cdot \vec{\sigma})$, where $\vec{v} \in \mathds{R}^3$ is required to have a length equal to or smaller than $1$ in order for $\rho$ to be positive semidefinite, and $\vec{\sigma}=(\sigma_x, \sigma_y, \sigma_z)^T$, with $\sigma_i$ being the Pauli matrices.

Any quantum channel corresponds to an affine transformation $\vec{v} \mapsto \Lambda\vec{v} + \vec\kappa$ with a real matrix $\Lambda$ and a displacement vector $\vec{\kappa}$ \cite{bengtsson2017geometry}, where some restrictions on $\Lambda$ and $\vec{\kappa}$ apply to ensure complete positivity. Thus, the image of any single-qubit channel $\mathcal{M}$ is given by an ellipsoid in the Bloch sphere, where the semiaxes are given by the singular values of $\Lambda$ and the ellipsoid is translated by $\vec\kappa$. The surface is given by the image of the pure states under $\mathcal{M}$ because of linearity (see Fig.~\ref{fig:qubitmaps}). Any maximally coherent state is a pure state, and vice versa, any pure state is maximally coherent in some basis. Since any transformation of $\vec{v}$ can be decomposed into rotations, contractions and a translation, the set of maximally coherent states in a fixed basis, forming a great circle in the Bloch picture, is mapped onto the boundary of an 
ellipse given by a cut through the center of the ellipsoid.

To find bounds on the quality of single-qubit M\&P channels, we use a geometric approach. 
$Q_-(\mathcal{M})$ determines the axis in the Bloch sphere and the ellipse on the image's surface of $\mathcal{M}$ that minimize the maximal distance of any point on this ellipse to the axis. This is because, in the computational basis, $C_\one(\rho) = |v_x + i v_y| = \sqrt{v_x^2 + v_y^2}$ \cite{napoli2016robustness}, which is the distance of a point at $\vec{v}$ from the $z$ axis, which defines the computational basis. For any other basis, the Bloch sphere can simply be rotated, leading to the same geometric result for any basis. For $Q_0(\mathcal{M})$, the ellipse is fixed by the axis depending on the channel $\mathcal{M}$. To find an upper bound on the measures $Q_-$ and $Q_+$ (and from the latter for $Q_0$), it is sufficient to replace the minimization over all axes by a fixed set of directions in the Bloch sphere, which allows us to obtain the following bounds.

\begin{figure}[t]
 \centering
 \includegraphics[width=0.4\columnwidth]{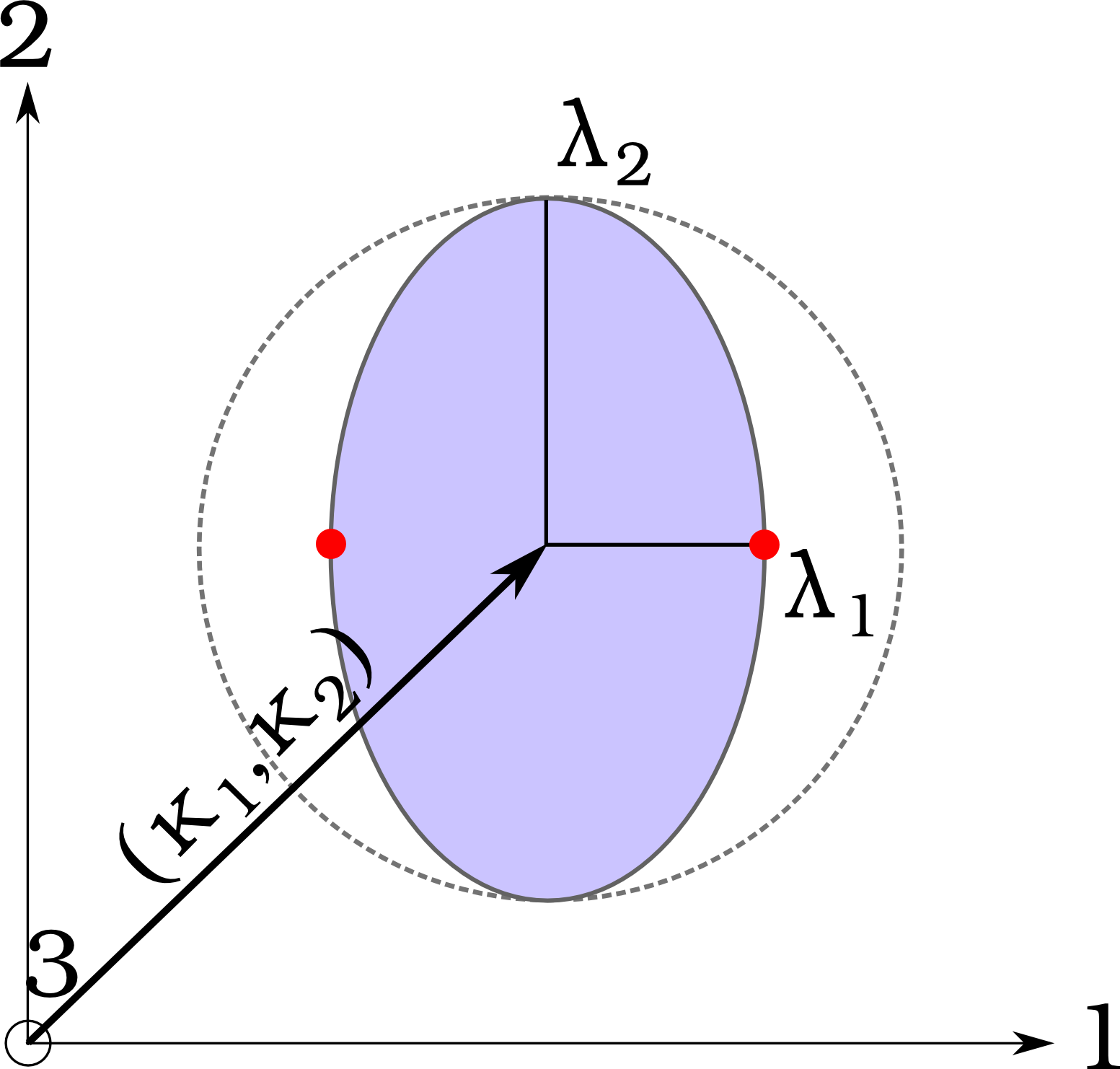}\quad\includegraphics[width=0.4\columnwidth]{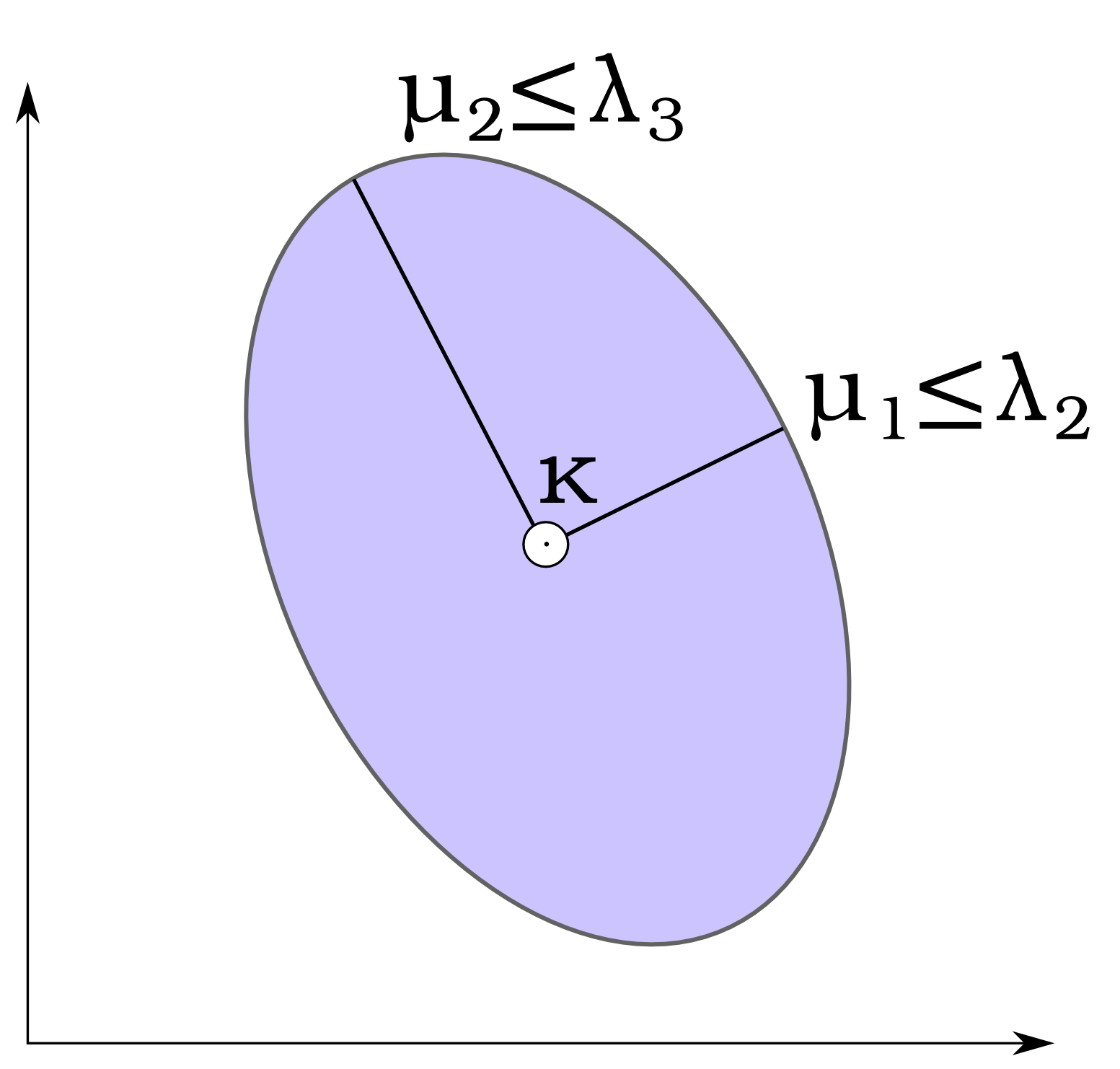}
 \caption{Left: Projection of the ellipsoid in the 1-2-direction to obtain upper bounds on the measures. Red dots indicate the points of the image of maximally coherent states in some basis which touch the boundary of the projected ellipse.
 Right: Projection of the ellipsoid in the direction of $\vec{\kappa}$. The semiaxes of the projection are bounded by the semiaxes of the ellipsoid.}
 \label{fig:qubitmaps2}
\end{figure}

\begin{lemma} \label{Lemma:upperbounds}
Let $\mathcal{M}$ be a single-qubit channel defined by displacement vector $\vec{\kappa}$ and transformation matrix $\Lambda$ with singular values $\lambda_1 \leq \lambda_2 \leq \lambda_3$. 
Let $\vec{\kappa} = (\kappa_1,\kappa_2,\kappa_3)^T$ in the bases where $\Lambda = \operatorname{diag}(\lambda_1, \lambda_2, \lambda_3)$. Then, $Q_-(\mathcal{M}) \leq \min(\sqrt{\kappa_1^2 + \kappa_2^2} + \lambda_1, \lambda_2)$ and $Q_0(\mathcal{M}) \leq Q_+(\mathcal{M}) \leq \min(\sqrt{\kappa_1^2 + \kappa_2^2} + \lambda_2, \lambda_3)$.
\end{lemma}
\begin{proof}
Instead of minimizing over all bases, we restrict the minimization to a discrete set to obtain an upper bound. For both $Q_-(\mathcal{M})$ and $Q_+(\mathcal{M})$, we consider the axes along $\vec{\kappa}$ and along the largest singular value of $\Lambda$.

To obtain an upper bound on $Q_+(\mathcal{M})$, we simply take into account all states on the surface of the ellipsoid. The largest possible distance to the axis along $\vec{\kappa}$ clearly is $\lambda_3$ since the axis goes through the center of the ellipsoid (see Fig.~\ref{fig:qubitmaps2}, right). Similarly, the distance from the axis along $\lambda_3$ is the distance to the center, which is given by $\sqrt{\kappa_1^2 + \kappa_2^2}$, plus at most $\lambda_2$ since the axis is parallel to $\lambda_3$ (see Fig.~\ref{fig:qubitmaps2}, left). Because of the minimization over all bases, an upper bound is then given by $\min(\sqrt{\kappa_1^2 + \kappa_2^2} + \lambda_2, \lambda_3)$.

In the case of $Q_-$, we can additionally choose the set of maximally coherent states. Since the channel $\mathcal{M}$ corresponds to an affine transformation of the Bloch vector, any ellipse on the surface of the ellipsoid with the same center as the ellipsoid is the image of a great circle on the surface of the Bloch sphere. Each of these circles is the set of maximally coherent states with respect to some basis. Hence, we can choose any ellipse on the surface of the ellipsoid and determine the maximal distance to the chosen axis to obtain an upper bound. For the axis along $\vec{\kappa}$, we choose the ellipse with semiaxes $\lambda_1$ and $\lambda_2$. Then, the maximal distance is at most $\lambda_2$ since the axis goes through the center of the ellipse. In the case of the axis along $\lambda_3$, the ellipse with semiaxes $\lambda_1$ and $\lambda_2$ limits the maximal distance to $\sqrt{\kappa_1^2 + \kappa_2^2} + \lambda_1$ (see Fig.~\ref{fig:qubitmaps2}, left). Again, the minimum of the cases 
considered gives an upper bound on $Q_-(\mathcal{M})$.
\end{proof}

One can also find lower bounds on the quantities, which will later be useful for applications.
\begin{lemma} \label{Lemma:lowerbounds}
Let $\mathcal{M}$ be a single-qubit channel defined by displacement vector $\vec{\kappa}$ and transformation matrix $\Lambda$ with singular values $\lambda_1 \leq \lambda_2 \leq \lambda_3$.
Then, $Q_0(\mathcal{M}) \geq Q_-(\mathcal{M}) \geq \lambda_1$ and $Q_+(\mathcal{M}) \geq \lambda_2$.
If $\mathcal{M}$ is unital ($\vec{\kappa} = 0$), equality holds for $Q_\pm$.
\end{lemma}
\begin{proof}
In order to find lower bounds, we have to show the bound in all coherence bases.

For $Q_+$, we have to consider -- for every coherence basis -- the maximal distance to the center of the projection of the ellipsoid onto the plane perpendicular to the coherence direction. This projection is an ellipse with semiaxes $\mu_1 \geq \lambda_1$ and $\mu_2 \geq \lambda_2$, displaced by some vector from the center. If the displacement is 0, the maximal distance is given by $\mu_2$ and therefore at least $\lambda_2$. For nonvanishing displacement, the maximal distance can only increase, yielding the lower bound for $Q_+$.

For $Q_-(\mathcal{M})$, we additionally have to minimize the maximal distance to the axis of two opposite points on this ellipse, due to the additional minimization over the input coherent states. This is in any case larger than $\mu_1$ and therefore larger than $\lambda_1$.

Finally, if the channel is unital, note that the minimum over the coherence bases is attained in the direction of $\lambda_3$, where for $Q_-(\mathcal{M})$, we consider the states mapped to an ellipse along the $\lambda_1$-$\lambda_3$-axes, giving a maximum distance of $\lambda_1$. For $Q_+(\mathcal{M})$, the maximum distance of the non-displaced ellipsoid in this basis is given by $\lambda_2$.
\end{proof}

The upper bounds on the quality measures can be used to obtain tight bounds for M\&P qubit channels.

\begin{figure}[t]
 \centering
 \includegraphics[width=0.49\columnwidth]{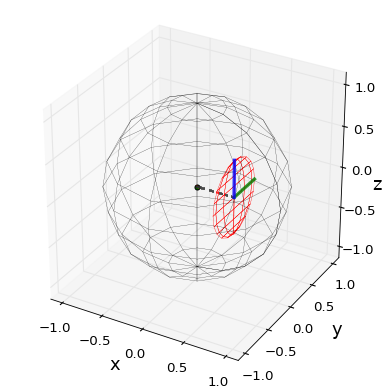} \includegraphics[width=0.49\columnwidth]{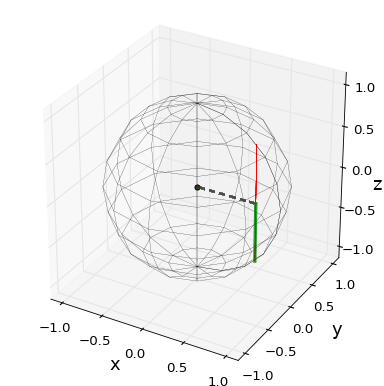} 
 \caption{Ellipsoid representations of the M\&P channels that maximize the different quality measures. The displacement vector $\vec{\kappa}$ is depicted by the dotted black line; the semiaxes, by the blue and green lines. Left: The M\&P channel maximizing $Q_-$ maps to a disk of radius $\frac1{\sqrt5}$, displaced by $\frac1{\sqrt5}$. Right: The M\&P channel maximizing $Q_0$ and $Q_+$, mapping to a straight line of length $\frac2{\sqrt2}$, displaced by $\frac1{\sqrt2}$.}
 \label{fig:sqrt5channel}
\end{figure}

\begin{theorem}\label{theorem:qubitbounds}
Let $\mathcal{M}$ be a single-qubit M\&P channel. Then, it holds that
\begin{align}
Q_0(\mathcal{M}) \leq Q_+(\mathcal{M}) \leq \frac{1}{\sqrt{2}}
\end{align}
and $Q_-(\mathcal{M}) \leq \frac{1}{\sqrt{5}}$. Additionally, if $\mathcal{M}$ is unital ($\vec{\kappa} = 0$),
\begin{align}
Q_0(\mathcal{M})\leq Q_+(\mathcal{M}) \leq \frac{1}{2}
\end{align}
and $Q_-(\mathcal{M}) \leq \frac{1}{3}$. All of these bounds are tight.
\end{theorem}
\begin{proof}
Let $\mathcal{M}$ be defined by displacement vector $\vec{\kappa}$ and transformation matrix $\Lambda$ with singular values $\lambda_1 \leq \lambda_2 \leq \lambda_3$. Since we only consider $Q_+$ and $Q_-$, we can assume without loss of generality~that $\Lambda = \operatorname{diag}(\lambda_1, \lambda_2, \lambda_3)$. Let $\vec\lambda = (\lambda_1, \lambda_2, \lambda_3)^T$.
Complete positivity of a single-qubit channel $\mathcal{M}$ is equivalent to $\rho_\mathcal{M} \geq 0$ where $\rho_\mathcal{M}$ is the Choi matrix of $\mathcal{M}$ \cite{geometrychan, pillis1967linear, jamiolkowski1972linear, choi1975completely}.
Using Descartes's rule of signs \cite{descartes} on the characteristic polynomial of the Choi matrix $\rho_\mathcal{M}$, complete positivity of the channel is equivalent to the set of inequalities
\begin{align}
|\vec{\kappa}|^2 + |\vec{\lambda}|^2 &\leq 3,\label{positivity1}\\ 
|\vec{\kappa}|^2 + |\vec{\lambda}|^2 - 2\lambda_1 \lambda_2 \lambda_3 &\leq 1,\\
\begin{split}
(1-|\vec{\kappa}|^2)^2 - 2(1 - |\vec{\kappa}|^2) |\vec{\lambda}|^2 -\frac12 |\vec{\lambda}|^4 \\ 
+ 8\lambda_1 \lambda_2 \lambda_3 +\frac12\sum_i D_i^2 -4\vec{K}\cdot\vec{L}& \geq0,\label{positivity3}
\end{split}
\end{align}
where $D_i = \sum_{j=1}^3 (-1)^{\delta_{ij}} \lambda_i^2$, $\vec{K}=(\kappa_1^2,\kappa_2^2,\kappa_3^2)^T$ and $\vec{L} = (\lambda_1^2,\lambda_2^2,\lambda_3^2)^T$.
Similarly, single-qubit channels are M\&P channels if and only if $\frac{1}{2} \one - \rho_\mathcal{M}$ is positive semidefinite \cite{geometrychan}. This yields the same set of equations with $\lambda_i \leftrightarrow -\lambda_i$.
In the following, we apply these restrictions to Lemma \ref{Lemma:upperbounds}. Clearly, the bounds from Lemma \ref{Lemma:upperbounds} only become worse if $\vec{\kappa}$ is rotated such that $\vec{\kappa} = (|\vec{\kappa}|, 0, 0)^T$. However, rotating a M\&P channel in such a way always leads to another M\&P channel as can be seen from Eqs.~(\ref{positivity1}) to (\ref{positivity3}). Thus, we can restrict ourselves to this type of channel.
For these channels, the eigenvalues can be evaluated analytically and maximization of the bounds over these channels for $Q_-$ results in the channel
\begin{align}
 \mathcal{M}_-(\rho) = \frac12\left[\one + \frac1{\sqrt5}\left(\sigma_x + \trace(\rho \sigma_y)\sigma_y + \trace(\rho \sigma_z)\sigma_z\right)\right].
\end{align}
It is visualized in the Bloch picture in Fig.~\ref{fig:sqrt5channel} and has the quality of $Q_-(\mathcal{M}_-) = \frac1{\sqrt5}$.
For $Q_+$, the optimization of the bounds over the channels yields
\begin{align}
 \mathcal{M}_+(\rho) = \frac12\left[\one + \frac1{\sqrt2}(\sigma_x + \trace(\rho \sigma_z)\sigma_z)\right],
\end{align}
with $Q_0(\mathcal{M}_+) = Q_+(\mathcal{M}_+) = \frac1{\sqrt2}$. The channel is visualized in Fig.~\ref{fig:sqrt5channel}.

For unital channels, i.e.,~$\vec{\kappa} = 0$, the condition for separability reads $\sum_i |\lambda_i| \leq 1$ \cite{entbrchan}. Maximizing under this constraint yields for $Q_-$ the depolarizing channel
\begin{align}
 \mathcal{M}_-^\prime(\rho) = \frac13\rho + \frac13 \one
\end{align}
with $Q_-(\mathcal{M}_-^\prime) = \frac13$.
For $Q_0$ and $Q_+$, we obtain the planar channel
\begin{align}
 \mathcal{M}_+^\prime(\rho) = \frac12\left[\one + \frac12(\trace(\rho \sigma_y)\sigma_y + \trace(\rho \sigma_z)\sigma_z)\right]
\end{align}
with $Q_0(\mathcal{M}_+^\prime) = Q_+(\mathcal{M}_+^\prime) = \frac12$.
\end{proof}

Finally, we have a statement similar to Lemma~\ref{lemma:prepqplus} for $Q_-$ if the channel is unital:
\begin{lemma}\label{lemma:singlequbitunitalprep}
Let $\mathcal{M}$ and $\mathcal{N}$ be unital channels acting on single qubits ($D=2$). Then, it holds that $Q_-(\mathcal{M}\circ \mathcal{N}) \leq Q_-(\mathcal{M})$.
\end{lemma}
\begin{proof}
First, note that the composition of unital channels is again a unital channel.
As shown in Lemma \ref{Lemma:lowerbounds}, the quality measure $Q_-(\mathcal{M})$ for a unital channel $\mathcal{M}$ is given by the minimal singular value of the matrix $\Lambda_\mathcal{M}$, i.e.,~$\lambda_1(\Lambda_\mathcal{M})$.  With this, we have that 
\begin{align}Q_-(\mathcal{M}\circ \mathcal{N}) &= \lambda_1(\Lambda_{\mathcal{M}\circ \mathcal{N}}) \nonumber\\
&\leq  \lambda_1(\Lambda_\mathcal{M})\lambda_3(\Lambda_\mathcal{N})\nonumber \\
&\leq  \lambda_1(\Lambda_\mathcal{M})=Q_-(\mathcal{M}).
\end{align}
For the first inequality, we have used the fact that $\Lambda_{\mathcal{M}\circ \mathcal{N}}=\Lambda_{\mathcal{M}}\Lambda_{\mathcal{N}}$ and Theorem 3.3.16 from Ref.~\cite{HornJohnsonTopicsMatrixAnalysis}. The second inequality follows from the fact that for channels, all the singular values of the matrix $\Lambda$ have to be less than or equal to~$1$.
\end{proof}

To illustrate how the measures can be determined for specific single-qubit channels, we examine several well-known channels. 

%-------------------------------------------------------------------------------------%
\section{Examples of single-qubit channels}
%-------------------------------------------------------------------------------------%
In the following, we consider the phase-flip, the amplitude-damping and the depolarizing channel and derive their quality in terms of $Q_0$ and $Q_\pm$ .

%%%%%%%%%%%%%%%%%%%%%%%%%%%%%%%%%%%%%%%%%%%
{\it -- The phase-flip channel $\mathcal{P}$:}
The matrix $\Lambda$ for the unital (i.e.,~$\vec{\kappa}=0$) phase-flip channel $\mathcal{P}$, is given by $\operatorname{diag}(1-p,1-p,1)$ with $0\leq p\leq 1$. It can be realized by a M\&P scheme for $p=1$ only.
Using the result from Lemma~\ref{Lemma:lowerbounds} for unital channels, we have that $Q_-(\mathcal{P}) = Q_0(\mathcal{P}) = Q_+(\mathcal{P})= 1-p$.
It should be noted that any bit-flip or bit-phase-flip channel is related to a phase-flip channel with the same error probability $p$ via a transformation of the form $\mathcal{V} \circ \mathcal{P}\circ \mathcal{V}^{-1}$,  where $\mathcal{V}(\rho) = V \rho V^\dagger$ is a unitary channel. Hence, the quality measures $Q_\pm$ and $Q_0$ for these channels with the same error probability coincide. Note that $Q_-$ excludes unital M\&P schemes for $p<\frac23$, while $Q_+$ and $Q_0$ exclude them for $p<\frac12$.

%%%%%%%%%%%%%%%%%%%%%%%%%%%%%%%%%%%%%%%%%%%
{\it -- The amplitude-damping channel $\mathcal{A}$:}
The matrix  $\Lambda$ for the amplitude-damping channel $\mathcal{A}$ is given by $\operatorname{diag}(\sqrt{1-p},\sqrt{1-p},1-p)$ and $\vec{\kappa}=(0,0,p)^T$,  where $0\leq p\leq 1$. This channel can again be implemented by M\&P schemes only if $p=1$.
Considering the maximal coherence of the states in the image of this channel with respect to the computational basis shows that $Q_+(\mathcal{A})\leq \sqrt{1-p}$. Using that $\lambda_1\leq Q_-\leq Q_0\leq Q_+$ leads to $Q_-=Q_0= Q_+=\sqrt{1-p}$. Thus, $Q_-$ excludes M\&P schemes for $p<\frac45$, whereas $Q_+$ and $Q_0$ exclude them for $p<\frac12$.

%%%%%%%%%%%%%%%%%%%%%%%%%%%%%%%%%%%%%%%%%%%
{\it -- The depolarizing channel $\mathcal{D}$:}
The matrix  $\Lambda$ for the unital depolarizing channel $\mathcal{D}$ is given by $\operatorname{diag}(p,p,p)$,  where $0\leq p\leq 1$. This channel is an M\&P channel only if $p \leq \frac{1}{3}$.
Because of symmetry, it is clear that $Q_- = Q_0 = Q_+ = p$. 
Thus, $Q_-$ certifies the full range of non-M\&P channels if it is known that the channel is unital, while $Q_0$ and $Q_+$ exclude M\&P schemes in the case of $p > \frac12$.

%-------------------------------------------------------------------------------------%
\section{Experimental estimation of the quality of a quantum memory}
%-------------------------------------------------------------------------------------%
\begin{figure}[t]
 \centering
 \includegraphics[width=0.8\columnwidth]{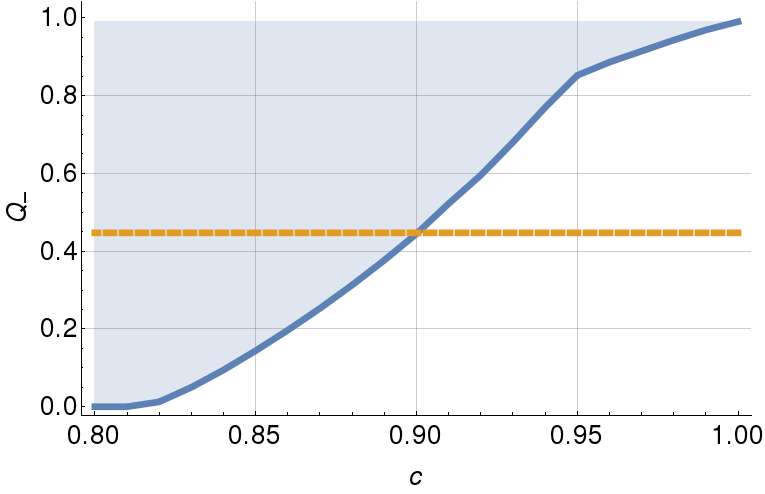}
 
 \caption{Lower bound (solid blue line) and allowed values above this bound (in blue) for the quality measure $Q_-$, given that in certain directions a coherence of at least $c$ is measured. The upper bound for M\&P channels of $\frac1{\sqrt5}$ is displayed by the dashed orange line. }
 \label{fig:l1bound}
\end{figure}
%%%%%%%%%%%%%%%%%%%%%%%
In this section, we explain how to determine a lower bound on the quality measures 
from experimental data for qubit systems for channels close to the identity channel. This situation is of major interest, as a perfect quantum memory corresponds to the identity channel.

Obviously, it is possible to obtain (lower bounds on) the quality measures by performing process tomography of the channel and then using the obtained characterization. However, process tomography requires the ability to prepare a set of input states with a high precision as well as many well characterized measurements \cite{nielsentomography, tomographyresource}. Here, we only assume that one can prepare three different states $\{\rho_i\}_{i=1}^3$ such that for the output states one can certify a lower bound $c_i\in [0,1]$ on the coherences
\begin{align}
C_{U_x} [\mathcal{M}(\rho_1)]\geq c_1, &\quad C_{U_y} [\mathcal{M}(\rho_1)]\geq c_1, \nonumber \\
C_{U_x} [\mathcal{M}(\rho_2)]\geq c_2, &\quad C_{U_z} [\mathcal{M}(\rho_2)]\geq c_2, \nonumber \\
C_{U_y} [\mathcal{M}(\rho_3)]\geq c_3, &\quad C_{U_z} [\mathcal{M}(\rho_3)]\geq c_3, \label{eqbc2}
\end{align}
where the $U_j$ correspond to the usual $x$, $y$ and $z$ directions on the Bloch sphere (i.e., $U_j=e^{i \sigma_j \pi/4}$ for $j=x,y,z$). 
This can, for instance, be achieved using the method from Ref.~\cite{yu2018detecting}. 
If the input states are chosen carefully and the channel is close enough to a unitary transformation, it suffices to conduct only three measurements in total.
These measurements certify that there are states close to the eigenstates of the Pauli matrices in the image of $\mathcal{M}$.
Furthermore, we only assume a bound on the coherence of the output of the quantum memory; 
nothing additional is assumed for the input- or output states.

For simplicity, we consider the case where $c:=c_1=c_2=c_3$. As the smallest semiaxis is a lower bound on $Q_-$, one can determine the channel that shows the smallest possible $\lambda_1$ compatible with the observed data. In particular, it is required that the image of the channel contains states for which the bounds given in Eqs.~(\ref{eqbc2}) are fulfilled. For $c>\sqrt{\frac23}\approx 0.82$, there must be at least three different states close to the boundary of the Bloch sphere. Numerically optimizing over all compatible channels leads to the lower bounds depicted in Fig.~\ref{fig:l1bound}.
Hence, for values of $c \gtrsim 0.82$ it is possible to obtain nontrivial lower bounds on the quality measure $Q_-$ (and, hence, also on $Q_0$ and $Q_+$) by having access to only a few lower bounds on the coherences of three different states. M\&P channels can be excluded with certainty if $Q_- > \frac1{\sqrt5} \approx 0.45$, which is given for $c\gtrsim 0.9$.

As an example, consider the amplitude-damping channel $\mathcal{A}$ from above. One can find states for which $c = \sqrt{1-p}$ and, thus, exclude M\&P channels for $p\lesssim 0.19$.\\

For higher-dimensional channels, the estimation is more involved.
In the following, we discuss how experimental data from higher-dimensional quantum memories $\mathcal{M}$ could be used to estimate the memory performance measure $Q_0(\mathcal{M})$. 
Since we know that $Q_0(\mathcal{M}) = 1$ iff $\mathcal{M}$ is unitary, we write $\mathcal{M} = \mathcal{V} + \mathcal{K}$, where $\mathcal{V}$ is some unitary channel and $\mathcal{K}(\rho) = \mathcal{M}(\rho) - \mathcal{V}(\rho)$ for all states $\rho$.
From Lemma 5, it follows that with respect to any basis $U$ with basis vectors $\ket{b_i}$, there always exist two maximally coherent states $\ket{\phi}$ and $\ket{\psi}$ such that $\mathcal{V}(\ket\phi) = \ket\psi$.
Thus, for $\rho = \mathcal{M}(\ket\phi) = \ket\psi\bra\psi + \mathcal{K}(\ket\phi\bra\phi)$ we have that
\allowdisplaybreaks
\begin{align}
C_U(\rho) &\geq - \frac{1}{D-1} \trace(W\rho) \nonumber \\
&= \frac{1}{D-1} \left( D \bra\psi \rho \ket\psi  - 1\right) \nonumber \\
&= \frac{1}{D-1} \left( D + D \bra\psi \mathcal{K}(\ket\phi\bra\phi) \ket\psi  - 1\right),
\end{align}
where we have used the notion of coherence witnesses introduced in Ref.~\cite{napoli2016robustness} with $W = \mathds{1} - D \ket\psi \bra\psi$, which gives a lower bound to the robustness of coherence.
Let $\lambda = \min_\sigma \lambda_{\min} [\mathcal{K}(\sigma)]$, i.e., $\lambda$ is the smallest eigenvalue of $\mathcal{K}(\sigma)$ for any state $\sigma$. 
Then,
\begin{align}
Q_0(\mathcal M) &= \min_U \max_{\vec{\alpha}}  C_U[\mathcal{M}(\ket{\Psi_U^{\vec{\alpha}}})] \nonumber \\
&\geq \frac{D(1+\lambda)-1}{D-1}.
\end{align}

To determine $\lambda$, we resort to the Choi matrix $\eta_{\mathcal K}$ of $\mathcal{K}$.
Hence
\begin{align}
  \lambda &= D \min_{\sigma,\ket s} \bra s \trace_A [(\sigma^T \otimes \mathds{1}) \eta_{\mathcal K}] \ket s \nonumber \\
          &= D \min_{\ket a \ket s} \bra a \bra s \eta_{\mathcal K} \ket a \ket s, 
\end{align}
 which can be estimated using experimental data.
For instance, let $\eta_{\mathcal M}$ and $\ket {\phi_\mathcal{V}} \bra {\phi_\mathcal{V}}$ be the Choi matrix of the cannels $\mathcal{M}$ and $\mathcal{V}$, respectively.  $\mathcal{V}$ might be guessed heuristically from the obtained data, determining $\ket {\phi_\mathcal{V}} \bra {\phi_\mathcal{V}}$.  The experimental data will impose linear constraints on $\eta_{\mathcal M}$ and, hence, also on $\eta_{\mathcal K}$. Using a see-saw optimization, it is possible to optimize $\lambda = D \min_{\sigma, \rho} \trace [(\sigma \otimes \rho) \eta_{\mathcal K}]$ over states $\sigma$ and $\rho$ with alternating semidefinite programs \cite{boyd}.

%-------------------------------------------------------------------------------------%
\section{Conclusions}
%-------------------------------------------------------------------------------------%
We introduced a physically motivated measure $Q_0$ that characterizes quantum memories
by their ability to preserve coherence. Using the upper and lower bound $Q_\pm$ we prove
that the measure fulfills all the desirable properties for such a quantifier. For a single-qubit quantum memory, the measure can be evaluated for many scenarios, even if only restricted experimental data is available. In contrast to full process tomography, our scheme does not require the precise preparation of states but demands only the certification of (sufficiently high) lower bounds on certain coherences of three unknown states.

For future work, it is desirable to extend the method to characterize and
verify other basic elements of quantum information processing. A simple extension
is the case of quantum teleportation, where the results can directly be applied.
More interesting is an application to two-qubit gates. The fact that a two-qubit 
gate generates entanglement, can be seen as the property that a certain two-level
coherence increases \cite{ringbauerpiani, kraftpiani}. In this sense, our method may be extended to characterize the 
entangling capability of multi-qubit quantum gates. 

{\it Note added:} Recently, we became aware of a similar approach, which was submitted recently \cite{shah2018quantumness}, also introducing a measure of quantum channels 
using coherence. Instead of considering the most robust or maximally coherent 
states, the authors are interested in the average coherence preserved over 
all states.

\begin{acknowledgments}
{\it Acknowledgments.---}
We thank Chau Nguyen and Tristan Kraft for discussions.
This work was supported by the DFG, the ERC (Consolidator 
Grant No. 683107/TempoQ), the Austrian Science Fund (FWF): J 4258-N27
and the House of Young Talents Siegen.
X.D.Y.~acknowledges funding from a CSC-DAAD scholarship.
\end{acknowledgments}

\section*{Appendix A: Proof that the measures are continuous}
Here, we prove the continuity of the quantities defined in the main text. To do so, we first show that the robustness of coherence is continuous. This settles a problem raised in Ref.~\cite{xi2019epsilon}.

\begin{lemma*}
The robustness of coherence is continuous.
\end{lemma*}
\begin{proof}
For a $D$-dimensional state $\rho$, the normalized robustness of coherence is given by
\begin{align}
	C_R(\rho) = \frac{1}{D-1} \min_{\sigma \in \mathcal{I}_U} \left\{ s \geq 0 \middle\vert \rho \leq (1+s) \sigma \right\},
\end{align}
where $\mathcal{I}_U$ is the set of incoherent $D$-dimensional states with respect to the basis defined by $U$ \cite{piani2016robustness}. Continuity means that for all $\epsilon > 0$ there exists a $\delta > 0$ such that for states $\rho$ and $\tau$ with $\left\Vert \rho - \tau \right\Vert < \delta$, it holds that $\left\vert C_R(\rho) - C_R(\tau)\right\vert < \epsilon$. 

We use the trace norm, which, for Hermitian matrices, is the sum of the absolute values of the eigenvalues. Thus, $\left\Vert \rho - \tau \right\Vert_{\trace} = \sum_j \left\vert \lambda_j(\rho-\tau) \right\vert < \delta$ implies that all eigenvalues of $\tau-\rho$ are upper bounded by $\delta$ and hence,
\begin{align}
\tau-\rho \leq \delta \mathds{1}. \label{taurho}
\end{align}
Let $C_R(\rho) = \frac{s^*}{D-1}$ and let $\sigma^* \in \mathcal{I}_U$ be a state such that $\rho \leq (1+s^*) \sigma^*$. 
Then, together with Eq.~(\ref{taurho}) we have that $\tau - \delta \mathds{1} \leq \rho \leq (1+s^*) \sigma^*$, or $\tau \leq [(1+s^*)+\delta D] \frac{(1+s^*)\sigma^* + \delta \mathds{1}}{(1+s^*)+\delta D}$. Since $\frac{(1+s^*)\sigma^* + \delta \mathds{1}}{(1+s^*)+\delta D}$ is a normalized incoherent state, it follows that $C_R(\tau) \leq \frac{s^* + \delta D}{D-1} = C_R(\rho) + \epsilon$, where we use that $C_R(\rho) = \frac{s^*}{D-1}$ and we choose $\delta = \frac{D-1}{D} \epsilon > 0$. Analogously, one finds that $C_R(\rho) \leq C_R(\tau) + \epsilon$ which completes the proof. 
\end{proof}

Using this result, we prove the continuity of our quantities.

\setcounter{definition}{12}
\begin{lemma*}
 The quantities $Q_\pm$ and $Q_0$ are continuous.
\end{lemma*}
\begin{proof}
  Let $\mathcal{M}$ be a quantum channel; then the corresponding Choi state $\eta_{\mathcal{M}}$ is given by \cite{pillis1967linear, jamiolkowski1972linear, choi1975completely}
  \begin{align}
   \eta_{\mathcal{M}} = \one \otimes \mathcal{M}(\ket{\phi^+}),
  \end{align}
  with $\ket{\phi^+}$ being the maximally entangled state $\frac{1}{\sqrt{D}}\sum_i \ket{ii}$.
  Using the Choi state, the inner part of the expressions for $Q_-$ and $Q_0$ can be written as
   \begin{align}
   \begin{split}
       & C_{U^\prime}(\mathcal{M}(\ket{\Psi_U^{\vec{\alpha}}})) \\
   = & C_{U^\prime}(D\,\trace_A [(\ketbra{\Psi_U^{\vec{\alpha}}}{\Psi_U^{\vec{\alpha}}}^T \otimes \one) \eta_\mathcal{M} ]),
   \end{split}
  \end{align}
  which is continuous in $\alpha, U, U^\prime$ and $\eta_\mathcal{M}$. Repeatedly applying the maximum theorem \cite{berge1997topological}, and using the fact that the robustness of coherence is continuous, shows that $Q_-$ and $Q_0$ are continuous in $\eta_\mathcal{M}$.

  If a sequence of channels $\{\mathcal{M}_i\}_i$ converges to a channel $\mathcal{M}$ with regard to~the diamond norm, then the sequence $\{\eta_{\mathcal{M}_i}\}_i$
  must converge to $\eta_\mathcal{M}$ \cite{kitaev1997quantum, watrous2018theoryofquant}.\footnote{Note that in finite dimensional systems, all norms and topologies are equivalent.}
  This implies that the function above is also continuous in $\mathcal{M}$. For $Q_+$, a similar argument holds. 
\end{proof}

\section*{Appendix B: Proof of Theorem 6}
In this Appendix, we prove Theorem 6 from the main text.

\setcounter{definition}{5}
\begin{theorem}
$Q_\pm$ and $Q_0$ fulfill property M1', i.e., if $Q(\mathcal{M})=1$, then 
$\mathcal{M}$ is a unitary channel.
\end{theorem}
\begin{proof}
To prove the theorem, it is sufficient to consider $Q_+(\mathcal{M})=1$, as $Q_-(\mathcal{M})\leq Q_0(\mathcal{M})\leq Q_+(\mathcal{M})\leq 1$. If  $Q_+(\mathcal{M})=1$, then for all unitaries $U$ it holds that
\begin{equation}
\max_{\ket{\psi}}  C_U[\mathcal{M}(\ket{\psi})]=1.
\end{equation}
This implies that for all $U$, there exists a state $\ket{\Phi}$ and a 
maximally coherent state $\ket{\Psi}$ with regard to~$U$ such that 
\begin{equation}
\label{eqQ1}\mathcal{M}(\ket{\Phi})=\ket{\Psi}.
\end{equation}

To prove the statement, we show the following three facts:
(i) If $Q_+(\mathcal{M})=1$, then we find a basis $\{\ket{\Phi_i}\}$ that 
is mapped to a basis $\{\ket{\Psi_i}\}$ by $\mathcal{M}$.
(ii) In the range of $\mathcal{M}$, there exist vectors 
$\{\ket{\Psi_{1j}} = \sum_{i=1}^D \beta_i^{(j)} \ket{\Psi_i}\}_{j=2}^D$ 
with the property $\beta_1^{(j)} \neq 0 \neq \beta_j^{(j)}$ for all $j$.
(iii) From the existence of the $\ket{\Psi_i}$ and $\ket{\Psi_{1j}}$, it follows that $\mathcal{M}$ is unitary.

For the first fact, in order to find state $\ket{\Psi_1}$, we simply 
choose a random basis 
and obtain a pure (maximally coherent) state in the range of $\mathcal{M}$ 
due to the property $Q_+(\mathcal{M})=1$. For the second state $\ket{\Psi_2}$, 
we choose a basis with $\ket{\Psi_1}$ as a basis state. The corresponding 
maximally coherent state has an overlap of 
$\vert\braket{\Psi_1|\Psi_2}\vert=\frac1{\sqrt{D}}$ 
and is therefore linearly independent. All other states $\ket{\Psi_i}$ can be 
found step by step: Let us assume that we have already found  
$\ket{\Psi_1}, \dots, \ket{\Psi_m}$ linearly independent states. 
We construct an orthonormal set of states spanning the same subspace 
and extend it to an orthonormal basis. The corresponding maximally 
coherent state has nonvanishing overlap with the space orthogonal to $\operatorname{span}\{\ket{\Psi_1}, \dots, \ket{\Psi_m}\}$ and is 
therefore also linearly independent.

With this procedure we obtain the nonorthonormal basis $\{\ket{\Psi_i}\}$. 
The corresponding preimages also form a basis, as,
from the Kraus decomposition (see also below) it follows that the 
dimension of their span must be equal to $D$ as well. 

For the second fact,  we  have to show the existence of the vectors 
$\{\ket{\Psi_{1j}}\}$ with 
the properties mentioned above. It suffices to show the existence of 
$\ket{\Psi_{12}}$; the proof for the other $\ket{\Psi_{1j}}$ is analogous.

Given the basis $\{\ket{\Psi_i}\}$, we consider the normalized dual basis $\{\ket{\gamma_i}\}$ with the property $\braket{\gamma_i|\Psi_j} = c_i \delta_{ij}$ 
for some $c_i > 0$ \cite{spiegel1959vector}. In this basis, 
$\beta_i^{(j)} = c_i^{-1}\braket{\gamma_i|\Psi_{1j}}$ holds. Now we search for
a vector $\ket{\Psi_{12}}$ in the range of $\mathcal{M}$ with the properties $\braket{\gamma_1|\Psi_{12}} \neq 0 \neq \braket{\gamma_2|\Psi_{12}}$, as from 
these conditions the presence of the desired coefficients $\beta_1^{(2)}$ and 
$\beta_2^{(2)}$ follows.

To this end, consider the orthonormal basis $\ket{b_1} = \ket{\gamma_1}$, 
$\ket{b_2} \propto \ket{\gamma_2} - \braket{\gamma_1|\gamma_2} \ket{\gamma_1}$ 
and the other $\ket{b_i}$ arbitrary.
The maximally coherent state $\ket{\Psi}$ in the range of $\mathcal{M}$ in this 
basis can be written as $\ket{\Psi}=\frac1{\sqrt D}\sum_{k=1}^D e^{i\phi_k} \ket{b_k}$. 
The overlaps are given by
\begin{align}
\braket{\gamma_1|\Psi} & \propto e^{i\phi_1} \neq 0, \nonumber\\
\braket{\gamma_2|\Psi} & \propto \braket{\gamma_2|b_1}e^{i\phi_1} + \braket{\gamma_2|b_2}e^{i\phi_2}.
\end{align}
If $\vert\braket{\gamma_2|b_1}\vert \neq \vert\braket{\gamma_2|b_2}\vert$, $\ket{\Psi_{12}} = \ket{\Psi}$ satisfies the desired properties.

Otherwise, we instead choose the basis $\ket{b_1^\prime} = \sqrt{\frac23} \ket{b_1} + \sqrt{\frac13} e^{i\theta} \ket{b_2}$ and $\ket{b_2^\prime} = \sqrt{\frac13} \ket{b_1} - \sqrt{\frac23} e^{i\theta} \ket{b_2}$ and the other $\ket{b_i'}$ arbitrary.
Now, the maximally coherent state $\ket{\Psi^\prime}=\frac1{\sqrt D}\sum_{k=1}^d e^{i\phi_k^\prime} \ket{b_k^\prime}$, with respect to the basis $\{\ket{b_i^\prime}\}$, in the range of $\mathcal{M}$ has the overlaps
\begin{align}
\braket{\gamma_1|\Psi^\prime}  \propto \,& \sqrt{\frac23} e^{i\phi_1^\prime} + \sqrt{\frac13} e^{i\phi_2^\prime} \neq 0,\\
\begin{split}
\braket{\gamma_2|\Psi^\prime}  \propto \,& (\sqrt{\frac23} e^{i\phi_1^\prime} + \sqrt{\frac13} e^{i\phi_2^\prime})\braket{\gamma_2|b_1} \\
                                       & + (\sqrt{\frac13} e^{i\phi_1^\prime} - \sqrt{\frac23} e^{i\phi_2^\prime})e^{i\theta}\braket{\gamma_2|b_2}.
\end{split}\label{eq:overlap2}
\end{align}

As in this case $\vert\braket{\gamma_2|b_1}\vert = \vert\braket{\gamma_2|b_2}\vert$, we can choose $\theta$ such that $\braket{\gamma_2|b_1} = e^{i\theta}\braket{\gamma_2|b_2} \neq 0$. 
Then the right-hand side of Eq.~(\ref{eq:overlap2}) is proportional to $(\sqrt2 + 1)e^{i\phi_1^\prime} + (1-\sqrt2)e^{i\phi_2^\prime}$, which cannot vanish. 
Thus, in this case we choose $\ket{\Psi_{12}} = \ket{\Psi^\prime}$.

Finally, concerning the third fact, as $\mathcal{M}$ is a quantum channel, it 
admits a Kraus representation, i.e.,~$\mathcal{M}(\rho)=\sum_{l=1}^r K_l \rho K_l^\dagger$
with $\sum_l K_l^\dagger K_l=\one$. Using the fact that the $\ket{\Phi_i}$ are mapped 
to pure states, we have for all $l = 1,\ldots, r$ that
\begin{align}\label{ko1}
K_l\ket{\Phi_i}&=\mu_{li} \ket{\Psi_i}
\end{align}
for $i=1,\ldots,D$, and
\begin{equation}
K_l\ket{\Phi_{1j}}=\kappa_{lj} \ket{\Psi_{1j}} \label{eq:Kl_phi1j}
\end{equation}
for some $\ket{\Phi_{1j}} = \sum_{k=1}^D \alpha_k^{(j)} \ket{\Phi_k}$ and $j=2,\ldots,D$.

Decomposing the right-hand side of Eq.~(\ref{eq:Kl_phi1j}) in terms of the basis $\{\ket{\Psi_i}\}$ and using linearity on the left-hand side, we have
\begin{align}
 \sum_{k=1}^D \mu_{lk} \alpha_k^{(j)} \ket{\Psi_k} = \kappa_{lj} \sum_{k=1}^D \beta_k^{(j)} \ket{\Psi_k}
\end{align}
for all $l$. Thus, for all $l,j$ and $k$, 
\begin{align}
 \mu_{lk} \alpha_k^{(j)} = \kappa_{lj} \beta_k^{(j)}.
\end{align}
For a fixed $j$, consider the two equations for $k=1$ and $k=j$, where the corresponding $\beta_k^{(j)}$ do not vanish by assumption.
If $\alpha_1^{(j)}$ or $\alpha_j^{(j)}$ were 0, $\kappa_{lj}=0$ for all $l$ would follow. This would imply that $\mathcal{M}(\ket{\Phi_{1j}})=0$, which cannot be true if $\mathcal{M}$ is a channel.

Otherwise, if $\kappa_{lj}$ was 0 for one $l$, then this would imply that $\mu_{l1} = \mu_{lj} = 0$ for this $l$. However, $\mu_{l1} = 0$ implies that $\kappa_{lj^\prime} = 0$ for all $j^\prime$, which in turn implies that $\mu_{lj^\prime}=0$ for all $j^\prime$. Thus, $K_l$ would map a whole basis to 0 and, therefore, vanishes and can be neglected from the decomposition of the channel.

Thus, we have that $\kappa_{lj}\neq0$ and, from that, $\mu_{l1}\neq0\neq\mu_{lj}$. Then the ratio
\begin{align}
 \frac{\mu_{l1}}{\mu_{lj}} = \frac{\beta_1^{(j)} \alpha_j^{(j)}}{\beta_j^{(j)}\alpha_1^{(j)}}
\end{align}
is independent of $l$. As this holds for all $j$, it follows from Eq.~(\ref{ko1}) that the $K_l$ must be proportional to each other, i.e.,~$K_l \propto K_{l^\prime}$. Using now that $\mathcal{M}$ is trace preserving, i.e.~$\sum_l K_l^\dagger K_l=\one$, leads to $K_l^\dagger K_l\propto \one$. Thus, all Kraus operators have to be proportional to the same unitary $V$, and hence, $\mathcal{M}(\rho) = V \rho V^\dagger$.
\end{proof}

It follows immediately from the proof above that, to completely characterize a unitary channel, it is sufficient to prepare a basis which is mapped to another basis by that channel and another pure state in the image which has nonvanishing coefficients in the latter basis. 
If one can find such states, the channel is guaranteed to be unitary and is uniquely determined by those states up to a global phase. 
Since pure states can be characterized with few measurements \cite{goyeneche2015five, carmeli2016stable}, the same can also be done with unitary quantum channels $\mathcal{M}$. 
Such a characterization has also been constructively obtained in Ref. \cite{baldwin2014quantum}.

\end{document}